\documentclass[11pt]{article}

\usepackage[hyphens]{url}
\usepackage[pagebackref,colorlinks]{hyperref}
\usepackage[hyphenbreaks]{breakurl}
\usepackage{amsmath} % AMS Math Package
\usepackage{amsthm} % Theorem Formatting
\usepackage{amssymb}	% Math symbols such as \mathbb
\usepackage{graphicx} % Allows for eps images
\usepackage{multicol} % Allows for multiple columns
\usepackage{multirow}
\usepackage{color}
\usepackage[dvips,letterpaper,margin=1in,bottom=1in]{geometry}
\usepackage[capitalize,noabbrev]{cleveref}

\usepackage{bm}
\usepackage{bbm}

\usepackage{ytableau}
\usepackage{longtable}

\usepackage[utf8]{inputenc}
\usepackage[english]{babel}

\usepackage[T1]{fontenc}
\AtBeginDocument{%
  \DeclareFontShape{T1}{cmr}{m}{scit}{<->ssub*cmr/m/sc}{}%
}

\usepackage{diagbox}
\usepackage{mathtools}

\newtheorem{theorem}{Theorem}[section]

\newtheorem{lemma}[theorem]{Lemma}

\newtheorem{fact}[theorem]{Fact}

\newtheorem{remark}{Remark}[section]

\DeclarePairedDelimiter\rbra{\lparen}{\rparen}
\DeclarePairedDelimiter\sbra{\lbrack}{\rbrack}
\DeclarePairedDelimiter\cbra{\{}{\}}
\DeclarePairedDelimiter\abs{\lvert}{\rvert}
\DeclarePairedDelimiter\Abs{\lVert}{\rVert}
\DeclarePairedDelimiter\ceil{\lceil}{\rceil}
\DeclarePairedDelimiter\floor{\lfloor}{\rfloor}

\DeclareMathOperator*{\E}{\mathbb{E}}
\DeclareMathOperator*{\Var}{\mathbb{V}\!\!\;ar}

\usepackage{latexsym}
\usepackage{CJK}

\usepackage{enumerate}
\usepackage{enumitem}

\usepackage{algorithm}
\usepackage{algpseudocode}
\usepackage{stmaryrd}
\usepackage{tabularx}
\usepackage{booktabs}
\usepackage{adjustbox}

\newcommand{\footremember}[2]{%
    \footnote{#2}
    \newcounter{#1}
    \setcounter{#1}{\value{footnote}}%
}

\usepackage{tikz}

\usetikzlibrary{quantikz2}

\begin{document}

\title{Settling the Sample Complexity of R\'enyi Entropy Estimation}
\author{Qisheng Wang \footremember{1}{Qisheng Wang is with the School of Computer Science, Shanghai Jiao Tong University (e-mail: \url{QishengWang1994@gmail.com}).}}
\date{}

\maketitle
\pagenumbering{roman}
\thispagestyle{empty}

\begin{abstract}
R\'enyi entropy estimation has been comprehensively investigated by \hyperlink{cite.AOST17}{Acharya, Orlitsky, Suresh and Tyagi (SODA 2015; \textit{IEEE Trans.\ Inf.\ Theory} 2017)} and consequent works, whereas only the sample complexity of R\'enyi entropy estimation of integer order has been settled. 
In this paper, we settle the sample complexity of R\'enyi entropy estimation of noninteger order, thereby completing the complexity picture of R\'enyi entropy estimation. 

Specifically, we show that for any noninteger $\alpha > 0$, it is sufficient and necessary to use
\[
\Theta\!\left(\frac{d^{\max\{1/\alpha,1\}}}{\varepsilon^{1/\alpha}\log(d)} + \frac{d^{|1-1/\alpha|}}{\varepsilon^2}\right)
\]
samples to estimate the R\'enyi entropy of order $\alpha$ of an unknown discrete distribution over an alphabet of size $d$ to within additive error $\varepsilon$. 
For the upper bound, we reduce the bias using a refined polynomial approximation estimator for large probabilities.
For the lower bound, we employ a different hard instance equipped with a new moment matching construction.
The constructive moment matching has constant bounded high-order moments, while attaining a fixed ratio between the $\alpha$-th moments, which is of independent interest. 
\end{abstract}

\newpage
\tableofcontents
\thispagestyle{empty}
\newpage
\pagenumbering{arabic}

\section{Introduction}

Entropy estimation is one of the fundamental questions in statistics and information theory not only because of its own interest but also due to numerous applications in different fields.
It was raised in \cite{Pan03} as an important problem from the perspective of algorithms and complexity, asking how many samples from an unknown probability distribution are needed to estimate its entropy. 
For a probability distribution $p = \rbra{p_1, p_2, \dots, p_d}$ over a finite alphabet of size $d$, its randomness is usually measured by the Shannon entropy \cite{Sha48a,Sha48b}
\[
\mathrm{H}\rbra{p} = - \sum_{i=1}^d p_i \log\rbra{p_i}. 
\]
The estimation of Shannon entropy has been extensively investigated in a series of works
\cite{Pan03,BDKR05,Pan04,Val11,VV11a,VV11b,VV17,JVHW15,JVHW17,WY16}. 
In particular, the sample complexity of Shannon entropy estimation was settled in \cite{JVHW15,WY16}: $\Theta\rbra{\frac{d}{\varepsilon\log\rbra{d}}+\frac{\log^2\rbra{d}}{\varepsilon^2}}$ samples are sufficient and necessary to estimate the Shannon entropy $\mathrm{H}\rbra{p}$ to within additive error $\varepsilon$ with high success probability. 
As a generalization, the estimation of functionals of the form $F\rbra{p} = \sum_{i=1}^d f\rbra{p_i}$ and multivariate functionals has also been extensively studied in the literature, e.g., \cite{HJW20,FS18,HJW18,WY20}.

As a popular generalization of Shannon entropy, the R\'enyi entropy of order $\alpha$ \cite{Ren61} is defined by
\[
\mathrm{H}_\alpha\rbra{p} = \frac{1}{1-\alpha} \log\rbra*{ \mathrm{F}_\alpha\rbra{p} }, \qquad \mathrm{F}_\alpha\rbra{p} = \sum_{i=1}^d p_i^\alpha,
\]
with the Shannon entropy $\mathrm{H}\rbra{p} = \mathrm{H}_1\rbra{p} = \lim_{\alpha \to 1} \mathrm{H}_\alpha\rbra{p}$ as the limiting case. 
R\'enyi entropy has been applied in unsupervised learning \cite{Xu98,JHE+03}, image processing \cite{MHGM00,SA04,NHZC06}, password guessing \cite{Ari96,PS04,HS11}, analysis of financial time series \cite{JKS12}, etc. 
The estimation of R\'enyi entropy has also been studied in a series of works \cite{Gra88,XE10,KLS12,AOST17,OS17,AK26}.
In particular, for integer $\alpha \geq 2$, it was shown in \cite{AOST17} that $\Theta_\alpha\rbra{\frac{d^{1-1/\alpha}}{\varepsilon^2}}$ samples are sufficient and necessary to estimate $\mathrm{H}_\alpha\rbra{p}$, the R\'enyi entropy of integer order $\alpha$, to within additive error $\varepsilon$, which was later improved to $\Theta\rbra{\frac{\alpha d^{1-1/\alpha}}{\varepsilon^2}}$ in \cite{AK26} with explicit $\alpha$-dependence for $2 \leq \alpha \lesssim \log\rbra{d}$.
For noninteger $\alpha > 0$ and $\alpha \neq 1$, the prior best results are due to \cite{AOST17}.
Specifically, for $0 < \alpha < 1$, $O\rbra{\frac{d^{1/\alpha}}{\varepsilon^{1/\alpha}\log\rbra{d}}}$ samples suffice for any sufficiently large $d$ relative to $\varepsilon$, whereas $\Omega\rbra{d^{1/\alpha-\gamma}}$ samples are necessary for any constant $\gamma > 0$.
For noninteger $\alpha > 1$, $O\rbra{\frac{d}{\varepsilon^{1/\alpha}\log\rbra{d}}}$ samples suffice for any sufficiently large $d$ relative to $\varepsilon$, whereas $\Omega\rbra{d^{1-\gamma}}$ samples are necessary for any constant $\gamma > 0$. 
It can be seen that there are still gaps between the upper and lower bounds on the sample complexity of R\'enyi entropy estimation of noninteger order: 
\begin{itemize}
    \item \textbf{Upper bounds:} The prior best result in \cite{AOST17} did not consider the case where $d$ is not sufficiently large relative to $\varepsilon$. 
    In that case, for example, when $\alpha > 1$, the sample complexity lower bound $\Omega\rbra{\frac{d^{1-1/\alpha}}{\varepsilon^2}}$ given in \cite{AOST17} applies, indicating that the $\varepsilon$-dependence should be refined, otherwise contradicting their upper bound $O\rbra{\frac{d}{\varepsilon^{1/\alpha}\log\rbra{d}}}$ since $1/\alpha < 2$. 
    A similar issue on the $\varepsilon$-dependence for $0 < \alpha < 1$ was also pointed out in \cite{WZL24}.
    \item \textbf{Lower bounds:} The lower bounds established in \cite{AOST17} does not fully match their upper bounds. The $\varepsilon$-dependence remains unclear and a possible (quasi-)polylogarithmic factor of $d$ in the sample complexity has not yet determined. 
\end{itemize}

In this paper, we close these gaps for all noninteger orders $\alpha$ from both sides by presenting optimal estimators together with matching lower bounds, thereby settling the sample complexity and completing the picture of R\'enyi entropy estimation.

\subsection{Main results}

We have the following matching upper and lower bounds on the sample complexity of R\'enyi entropy estimation. 

\begin{theorem}[\cref{thm:estimator,thm:lower-bound} combined] \label{thm:main}
    Fix a noninteger constant $\alpha > 0$ and $\alpha \neq 1$.
    For an unknown discrete probability distribution $p$ over a finite alphabet of size $d$, it is sufficient and necessary to use 
    \[
    \Theta\rbra*{ \frac{d^{\max\cbra{1/\alpha, 1}}}{\varepsilon^{1/\alpha} \log\rbra{d}} + \frac{d^{\abs{1-1/\alpha}}}{\varepsilon^2} } = \begin{cases}
        \Theta\rbra*{ \dfrac{d^{1/\alpha}}{\varepsilon^{1/\alpha} \log\rbra{d}} }, & 0 < \alpha \leq 1/2, \\
        \Theta\rbra*{ \dfrac{d^{1/\alpha}}{\varepsilon^{1/\alpha} \log\rbra{d}} + \dfrac{d^{1/\alpha - 1}}{\varepsilon^2} }, & 1/2 < \alpha < 1, \\
        \Theta\rbra*{\dfrac{d}{\varepsilon^{1/\alpha}\log\rbra{d}} + \dfrac{d^{1 - 1/\alpha}}{\varepsilon^2} }, & \alpha > 1, \alpha \notin \mathbb{N},
    \end{cases}
    \]
    samples from $p$ to estimate the R\'enyi entropy of order $\alpha$, $\mathrm{H}_\alpha\rbra{p}$, to within additive error $\varepsilon$ with success probability at least $2/3$. 
\end{theorem}

The sample complexity given in \cref{thm:main} is optimal up to a constant factor that depends only on $\alpha$. 
We collect the sample complexities of R\'enyi entropy estimation in \cref{tab:renyi}.

\begin{table}[!htp]
    \centering
    \caption{Sample complexity of R\'enyi entropy estimation.}
    \label{tab:renyi}
    \begin{tabular}{ccc}
        \toprule
        $\alpha$ & Sample Complexity & References \\
        \midrule
        $0 < \alpha \leq 1/2$ & $\Theta\rbra*{ \dfrac{d^{1/\alpha}}{\varepsilon^{1/\alpha} \log\rbra{d}} }$ & \cite{AOST17} \& This Work \\
        \midrule
        $1/2 < \alpha < 1$ & $\Theta\rbra*{ \dfrac{d^{1/\alpha}}{\varepsilon^{1/\alpha} \log\rbra{d}} + \dfrac{d^{1/\alpha - 1}}{\varepsilon^2} }$ & This Work \\
        \midrule
        $\alpha = 1$ (Shannon) & $\Theta\rbra*{\dfrac{d}{\varepsilon\log\rbra{d}}+\dfrac{\log^2\rbra{d}}{\varepsilon^2}}$ & \cite{JVHW15,WY16} \\
        \midrule
        $\alpha > 1, \alpha \notin \mathbb{N}$ & $\Theta\rbra*{\dfrac{d}{\varepsilon^{1/\alpha}\log\rbra{d}} + \dfrac{d^{1 - 1/\alpha}}{\varepsilon^2} }$ & This Work \\
        \midrule
        $\alpha \geq 2, \alpha \in \mathbb{N}$ & $\Theta\rbra*{\dfrac{d^{1-1/\alpha}}{\varepsilon^2}}$ & \cite{AOST17,AK26} \\
        \bottomrule
    \end{tabular}
\end{table}

We compare prior results with ours in the following. 

\paragraph{Comparison with \cite{AOST17}.}
Although the sample complexity given in \cite{AOST17} only considers the case where $d$ is sufficiently large relative to $\varepsilon$, i.e., their sample complexity holds for any $d \geq d_0\rbra{\alpha, \varepsilon}$, with a more refined analysis, one can show that the sample complexity of the (non-plug-in) estimator in \cite{AOST17} is
\begin{equation} \label{eq:complexity-AOST}
\begin{cases}
    O\rbra*{\dfrac{d^{1/\alpha}}{\varepsilon^{1/\alpha} \log\rbra{d}}}, & 0 < \alpha \leq 1/2, \\
    O\rbra*{ \min\cbra*{ \dfrac{d^{1/\alpha}}{\varepsilon^{1/\alpha} \log\rbra{d/\varepsilon}} + \exp\rbra*{\dfrac{C_\alpha}{\varepsilon}}, \dfrac{d^{1/\alpha}}{\varepsilon^{1/\alpha} \log^{1/\alpha-1}\rbra{d/\varepsilon}} } + \dfrac{d^{1/\alpha-1}}{\varepsilon^2} }, & 1/2 < \alpha < 1, \\
    O\rbra*{ \min\cbra*{ \dfrac{d}{\varepsilon^{1/\alpha} \log\rbra{d}} + \exp\rbra*{\dfrac{C_\alpha}{\varepsilon}}, \dfrac{d}{\varepsilon} } + \dfrac{d^{1-1/\alpha}}{\varepsilon^2} }, & \alpha > 1, \alpha \notin \mathbb{N},
\end{cases}
\end{equation}
where $C_\alpha > 0$ is a constant that depends on $\alpha$.\footnote{The sample complexity of the estimator in \cite{AOST17} was worked out through the use of Large Language Model.} 
In comparison, our upper bound improves theirs for all nonintegers $\alpha > 1/2$. 
The sample complexity lower bound given in \cite{AOST17} is
\[
\begin{cases}
    \Omega\rbra*{d^{1/\alpha - \gamma} + \dfrac{d^{1/\alpha-1}}{\varepsilon^2}}, & 0 < \alpha < 1, \\
    \Omega\rbra*{d^{1 - \gamma} + \dfrac{d^{1-1/\alpha}}{\varepsilon^2}}, & \alpha > 1, \alpha \notin \mathbb{N}.
\end{cases}
\]
In comparison, our lower bound improves theirs for all nonintegers $\alpha > 0$. 

\paragraph{Comparison with \cite{JVHW15}.}
Another R\'enyi entropy estimator for $1 < \alpha < 3/2$ was given in \cite[Theorem 5]{JVHW15} with sample complexity
\[
O\rbra*{\frac{d}{\varepsilon^{1/\rbra{\alpha-1}}\log\rbra{d/\varepsilon}}}.
\]
In comparison, our upper bound improves theirs for all $1 < \alpha < 3/2$.

\subsection{Techniques}

Let $p = \rbra{p_1, p_2, \dots, p_d}$ be an unknown discrete probability distribution.
We prove our results in the Poissonized sampling model, which is usually used in the functional estimation literature, e.g., \cite{JVHW15,WY16,AOST17}. 
Here, the standard sampling model is called the multinomial sampling model. 
In the \textit{multinomial sampling model}, $n$ samples from the distribution $p$, denoted as $\textsc{Multinomial}\rbra{p, n}$, are given as the input to the estimator, which is described by a sequence of $n$ independent and identical random variables $Y_1, Y_2, \dots, Y_n$ distributed from $p$. 
In the \textit{Poissonized sampling model}, (expected) $n$ samples from the distribution $p$, denoted as $\textsc{Poissonized}\rbra{p, n}$, are given as the input to the estimator, which is described by a sequence of $N$ independent and identical random variables $Y_1, Y_2, \dots, Y_N$ distributed from $p$, where $N$ obeys the Poisson distribution $\operatorname{Poi}\rbra{n}$.
Let $X_i = \sum_{j=1}^N \mathbf{1}_{\cbra{Y_j = i}}$ be the empirical count.
Then, it is noted that $X_1, X_2, \dots, X_d$ are independent and $X_i \sim \operatorname{Poi}\rbra{np_i}$. 
The following fact reveals that the sample complexity in the Poissonized sampling model is equivalent to the sample complexity in the multinomial sampling model only up to a constant factor.

\begin{fact}[{\cite[Lemmas 7 and 17]{AOST17}}, see also {\cite[Lemma 16]{JVHW15}} and {\cite[Appendix A]{WY16}}]
    Let $S_\alpha^{\textup{Mul}}\rbra{d, \varepsilon}$ and $S_\alpha^{\textup{Poi}}\rbra{d, \varepsilon}$ be the sample complexities of estimating the R\'enyi entropy of order $\alpha$ of an unknown $d$-dimensional discrete distribution $p$ to within additive error $\varepsilon$ in the multinomial sampling model and the Poissonized sampling model, respectively, with success probability at least $2/3$.
    Then, $S_\alpha^{\textup{Mul}}\rbra{d, \varepsilon} = \Theta\rbra{S_\alpha^{\textup{Poi}}\rbra{d, \varepsilon}}$. 
\end{fact}

\subsubsection{Lower bounds}

The lower bounds $\Omega\rbra{d^{1/\alpha-\gamma}}$ for $0 < \alpha < 1$ and $\Omega\rbra{d^{1-\gamma}}$ for noninteger $\alpha > 1$ for R\'enyi entropy estimation in \cite{AOST17} build on the existence of two nonnegative vectors $x, y \in \mathbb{R}^{m}$ for any $m \geq 2$ and noninteger $\alpha > 0$ such that $\Abs{x}_r = \Abs{y}_r$ for $1 \leq r < m$, $\Abs{x}_m \neq \Abs{y}_m$, and $\Abs{x}_\alpha \neq \Abs{y}_\alpha$. 
Let $S = \Abs{x}_1 = \Abs{y}_1$ for convenience. 
\begin{itemize}
    \item For noninteger $\alpha > 1$, the hard instance is a pair of distributions $\rbra{p^x, p^y}$ given by
    \begin{align*}
    p^x = \rbra[\Big]{ \underbrace{\frac{x_1}{DS}, \dots, \frac{x_1}{DS}}_{D \textup{ equal entries}}, \dots, \underbrace{\frac{x_m}{DS}, \dots, \frac{x_m}{DS}}_{D \textup{ equal entries}} },
    \end{align*}
    with $D = d/m$. 
    It can be shown that $\abs{\mathrm{H}_\alpha\rbra{p^x} - \mathrm{H}_\alpha\rbra{p^y}} = \Omega\rbra{1}$ whereas distinguishing $p^x$ and $p^y$ (after a uniformly random label permutation) requires $\Omega\rbra{d^{1-1/m}}$ samples, which gives the $\Omega\rbra{d^{1-\gamma}}$ lower bound since $m$ can be chosen arbitrarily large.

    \item For $0 < \alpha < 1$, the hard instance is a pair of distributions $\rbra{p^x, p^y}$ given by
    \begin{align*}
        p^x = \rbra[\Big]{1 - \frac{1}{D^\beta}, \underbrace{\frac{x_1}{D^{1+\beta}S}, \dots, \frac{x_1}{D^{1+\beta}S}}_{D \textup{ equal entries}}, \dots, \underbrace{\frac{x_m}{D^{1+\beta}S}, \dots, \frac{x_m}{D^{1+\beta}S}}_{D \textup{ equal entries}} },
    \end{align*}
    with $D = \rbra{d-1}/m$.
    It can be shown that $\abs{\mathrm{H}_\alpha\rbra{p^x} - \mathrm{H}_\alpha\rbra{p^y}} = \Omega\rbra{1}$ whereas distinguishing $p^x$ and $p^y$ requires $\Omega\rbra{d^{1/\alpha-1/m+\beta}}$ samples, which gives the $\Omega\rbra{d^{1/\alpha-\gamma}}$ lower bound since $m$ can be chosen arbitrarily large with $\beta = 1/\alpha - 1 - \eta/2$.
\end{itemize}

In this work, we present a different hard instance while unifying both the cases of $0 < \alpha < 1$ and noninteger $\alpha > 1$. 
To this end, our key technical lemma (\cref{lemma:mm}) is a constructive moment matching: for noninteger $\alpha > 0$, real $R > 1$, and integer $K \geq 2$, there is a pair of probability distributions $\rbra{\nu_0, \nu_1}$ on $\sbra{0, K^2}$ with $K$ matching moments and mean $\mu$ such that $\mu_1^{\rbra{\alpha}} \geq R \mu_0^{\rbra{\alpha}}$ while $\mu, \mu_b^{\rbra{2}}, \mu_b^{\rbra{\alpha}}, \mu_b^{\rbra{2\alpha}} \asymp_\alpha 1$, where $\mu_b^{\rbra{\alpha}} = \E_{X \sim \nu_b}\sbra{X^\alpha}$. 
It is worth noting that the constructed $\nu_0$ and $\nu_1$ (i) are discrete distributions over a finite set $\mathcal{X} = \cbra{0, \eta^r, \eta^{r-1}, \dots, \eta, 1^2, 2^2, \dots, K^2} \subseteq \sbra{0, K^2}$, and (ii) have constant bounded $2$-nd, $\alpha$-th, and $2\alpha$-th moments, which may be of independent interest. 

Our hard instance is then constructed as follows.
Consider the hypothesis $b \in \cbra{0, 1}$. 
Let $U_1, U_2, \dots, U_m \sim \nu_b$ be independent and identical random variables, where $\rbra{\nu_0, \nu_1}$ is chosen with $K \asymp \log\rbra{d}$ and $R \asymp_\alpha 1$. 
The probability distribution is defined by
\[
p_b = \rbra[\Big]{ \underbrace{tU_1, tU_2, \dots, tU_m}_{m}, \underbrace{\frac{R_U}{s}, \dots, \frac{R_U}{s}}_{s}, \underbrace{0, \dots, 0}_{d-m-s} }, \qquad R_U = 1 - t \sum_{i=1}^m U_i,
\]
where $m, s$ are chosen depending on $\alpha, d, \varepsilon$ and $t \asymp_\alpha \frac{\varepsilon^{1/\alpha}}{m^{1/\alpha} s^{1-1/\alpha}}$. 
Specifically, for $0 < \alpha < 1$, we choose $m = d-1$, $s = 1$, and $t \asymp_\alpha \varepsilon^{1/\alpha}/d^{1/\alpha}$; for noninteger $\alpha > 1$, $m = s = \floor{d/2}$ and $t \asymp_\alpha \varepsilon^{1/\alpha}/d$. 
In case $R_U < 0$, $p_b$ can be defined arbitrary, e.g., $p_b = \rbra{1, 0, \dots, 0}$. 
It can be shown that $\abs{\mathrm{H}_\alpha\rbra{p_0} - \mathrm{H}_\alpha\rbra{p_1}} > \varepsilon$ with high probability whereas distinguishing $p_0$ and $p_1$ requires $\Omega\rbra{\frac{1}{tK}} = \Omega\rbra{\frac{d^{\max\cbra{1/\alpha, 1}}}{\varepsilon^{1/\alpha}\log\rbra{d}}}$ samples, which gives the first term of the lower bound in \cref{thm:main}. 

\begin{remark}[Comparison with prior similar hard instances] \
\begin{itemize}
    \item \textbf{Comparison with \cite{WY16}.}
    A similar hard instance, though using different moment matching constructions, was also used for Shannon entropy estimation in \cite{WY16} with the special case of $s = 1$. 
    However, in our case, the choice of asymptotically large $s$ is necessary. 
    This is because in our construction, it is required to set $t \asymp_\alpha \frac{\varepsilon^{1/\alpha}}{m^{1/\alpha} s^{1-1/\alpha}}$.
    Therefore, a constant $s \asymp 1$ will only give a lower bound of $\frac{1}{tK} \asymp \frac{d^{1/\alpha}}{\varepsilon^{1/\alpha}\log\rbra{d}}$ for noninteger $\alpha > 1$, while $s \asymp d$ suffices to give a matching lower bound. 

    \item \textbf{Comparison with \cite{Wan26}.}
    Our hard instance, though using different moment matchings, is actually inspired by the recent lower bound hard instances for quantum property testing \cite{Wan26}, which can be understood as a diagonal quantum state.
    Here, we use a constructive moment matching for our purpose; in comparison, the moment matching used in \cite[Lemma 5.8]{Wan26} is not constructive (which relies on the Hahn-Banach theorem) and does not particularly care about bounds on high-order moments. 
\end{itemize}
\end{remark}

The second term of the lower bound in \cref{thm:main} simply follows from \cite{AOST17,Wan26}, as stated below. 

\begin{lemma}[{\cite[Theorems 15 and 26]{AOST17}} and {\cite[Lemma A.1]{Wan26}}]\label{lemma:lower-bound-partial}
    Let $\alpha > 0$ be a constant with $\alpha \neq 1$.
    For sufficiently large $d \geq 2$ and sufficiently small $\varepsilon > 0$, 
    \[
    S_\alpha\rbra{d, \varepsilon} = \Omega\rbra*{\frac{d^{\abs{1-1/\alpha}}}{\varepsilon^2}}.
    \]
\end{lemma}

\subsubsection{Upper bounds}

To estimate $\mathrm{H}_\alpha\rbra{p}$, the key step is to estimate the term $\mathrm{F}_\alpha\rbra{p}$ in the logarithm to a relative error. 
Let $X_i, Y_i \sim \operatorname{Poi}\rbra{np_i}$ be the empirical counts of two independent experiments, each with (expected) $n$ Poissonized samples. 
The estimator for $\mathrm{F}_\alpha\rbra{p}$ given in \cite{AOST17} is of the form
\begin{align*}
\widehat{F}^{\textup{AOST}} 
& = \sum_{i=1}^d \rbra*{ \mathbf{1}_{\cbra{Y_i \leq c_1L}} \mathsf{L}_n\rbra{X_i} + \mathbf{1}_{\cbra{Y_i > c_1L}} \mathsf{H}_n^{\textup{AOST}}\rbra{X_i} }, \\
\mathsf{L}_n\rbra{x} & = \sum_{j=0}^{c_2L} \frac{a_j \rbra{c_3L}^{\alpha - j}}{n^{\alpha}} x^{\underline{j}}, \qquad \mathsf{H}_n^{\textup{AOST}}\rbra{x} = \frac{x^\alpha}{n^\alpha},
\end{align*}
where $c_1, c_2, c_3 > 0$ are three universal constants, $c_1L$ with $L = \log\rbra{n}$ is the threshold that splits the empirical probabilities into two (small and large) parts, and $a_j$'s are the coefficients of the best approximation polynomial of degree $c_2 L$ of $x^\alpha$. 
Here, $x^{\underline{j}} = x \rbra{x-1} \dots \rbra{x-j+1}$ is the falling factorial and $x^{\underline{0}} = 1$.
It can be shown that the absolute bias of their estimator is 
\[
\abs*{ \E\sbra{\widehat{F}^{\textup{AOST}}} - \mathrm{F}_\alpha\rbra{p} } \lesssim_\alpha \frac{d}{\rbra{nL}^{\alpha}} + \frac{\mathrm{F}_\alpha\rbra{p}}{L}.
\]
Further analysis will give the sample complexity stated in \cref{eq:complexity-AOST} of the estimator in \cite{AOST17}.

In this work, we present an improved estimator of the form
\begin{align*}
    \widehat{F}^{\textup{our}} & = \sum_{i=1}^d \rbra*{ \mathbf{1}_{\cbra{Y_i \leq c_1L}} \mathsf{L}_n\rbra{X_i} + \mathbf{1}_{\cbra{Y_i > c_1L}} \mathsf{H}_n^{\textup{our}}\rbra{X_i, Y_i} }, \\
    \mathsf{H}_n^{\textup{our}}\rbra{x, y} & = \frac{1}{n^\alpha} \sum_{j=0}^J \binom{\alpha}{j} \sum_{r=0}^j \binom{j}{r} \rbra{-1}^{j-r} y^{\alpha-r} x^{\underline{r}}, \qquad J = \ceil{4\alpha} + 3.
\end{align*}
In comparison, we use the same estimator $\mathsf{L}_n$ for small probabilities, while using a different estimator $\mathsf{H}_n^{\textup{our}}$ for large probabilities. 
The intuition is that when $X \sim \operatorname{Poi}\rbra{\lambda}$, 
\[
\E\sbra*{ \mathsf{H}_n^{\textup{our}} \rbra{X, y} } = \frac{1}{n^\alpha} \sum_{j=0}^{J} \binom{\alpha}{j} \sum_{r=0}^j \binom{j}{r} \rbra{-1}^{j-r} y^{\alpha-r} \lambda^r = \sum_{j=0}^{J} \binom{\alpha}{j} y^{\alpha-j} \rbra{\lambda - y}^{j},
\]
which is the Taylor expansion of $x^{\alpha}$ at $y$ with degree $J$. 
This can be viewed as an estimator using the \textit{Taylor series bias correction} method (cf.\ \cite[Equation (86)]{JH20}); more precisely, it can be understood as a bias-corrected version of $\mathsf{H}_n^{\textup{AOST}}$.
In this way, the absolute bias can be reduced to
\[
\abs*{ \E\sbra{\widehat{F}^{\textup{our}}} - \mathrm{F}_\alpha\rbra{p} } \lesssim_\alpha \frac{d}{\rbra{nL}^{\alpha}},
\]
removing the unexpected term $\mathrm{F}_\alpha\rbra{p}/L$. 
Then, further analysis will show the sample complexity upper bound stated in \cref{thm:main}. 

\section{Preliminaries}

In this section, we list the results known in the literature that will be used later in this paper.

\subsection{Basic inequalities}

\begin{lemma}[Adapted from {\cite[Lemma 1]{AOST17}}] \label{lemma:2alpha-1-bound}
    For $\alpha > 1/2$ and $\alpha \neq 1$, $\mathrm{F}_{2\alpha-1}\rbra{p} / \rbra{\mathrm{F}_\alpha\rbra{p}}^2 \leq d^{\abs{1-1/\alpha}}$. 
\end{lemma}
\begin{proof}
    For $1/2 < \alpha < 1$, this is obtained by \cite[Lemma 1(iii)]{AOST17}: $\mathrm{F}_{\alpha-\beta}\rbra{p} \leq d^{\beta/\alpha} \rbra{\mathrm{F}_\alpha\rbra{p}}^{\rbra{\alpha-\beta}/\alpha}$ for $0 \leq \beta \leq \alpha$, with $\beta = 1-\alpha$.
    For $\alpha > 1$, this is obtained by \cite[Lemma 1(iv)]{AOST17}: $\mathrm{F}_{\alpha+\beta}\rbra{p} \leq d^{\rbra{\alpha-1}\rbra{\alpha-\beta}/\alpha} \rbra{\mathrm{F}_\alpha\rbra{p}}^2$, with $\beta = \alpha - 1$. 
\end{proof}

\subsection{Best polynomial approximation}

\begin{lemma} [Cf.\ {\cite[Lemmas 5 and 6]{AOST17}}] \label{lemma:poly-approx-xalpha}
    For every noninteger $\alpha > 0$, let $Q_k\rbra{x} = \sum_{j=0}^k a_j x^j$ be a best approximation polynomial of degree $k$ of $x^{\alpha}$ on $\sbra{0, 1}$, then 
    \[
    \Abs{Q_k - x^{\alpha}}_{\infty, \sbra{0, 1}} \leq c_\alpha k^{-2\alpha}, \qquad \sum_{j=0}^k \abs{a_j} \leq c_\alpha 8^k,
    \]
    where $c_\alpha \geq 1$ is a constant depending on $\alpha$. 
\end{lemma}
\begin{proof}
    Note that $\Abs{Q_k - x^\alpha}_{\infty, \sbra{0, 1}} = E_k\rbra{x^{\alpha}, \sbra{0, 1}} \leq C_\alpha k^{-2\alpha}$ for some constant $C_\alpha > 0$ for any integer $k \geq 1$ (cf.\ \cite[Theorem 8.1]{Lub07}).
    On the other hand, since $p\rbra{x} \equiv 1/2$ is a polynomial of degree $0$, $\Abs{Q_k - x^\alpha}_{\infty, \sbra{0, 1}} \leq \Abs{p - x^\alpha}_{\infty, \sbra{0, 1}} \leq 1/2$, we have $\Abs{Q_k}_{\infty, \sbra{0, 1}} \leq 3/2$. 
    By \cite[Lemma 2.8]{She20}, $\sum_{j=0}^k\abs{a_j} \leq 8^k \Abs{Q_k}_{\infty, \sbra{0, 1}} \leq \rbra{3/2} 8^k$. 
    Therefore, it suffices to take $c_\alpha = \max\cbra{3/2, C_\alpha}$. 
\end{proof}

\subsection{Inequalities for Poisson distributions}

\begin{lemma}[{\cite[Lemma 3]{WY16}}] \label{lemma:WY16-lemma3}
    Let $U$ and $V$ be random variables on $\sbra{0, M}$ such that $\E\sbra{U^j} = \E\sbra{V^j}$ for $j = 1, 2, \dots, L$.
    If $L > 2eM$, then
    \[
    d_{\textup{TV}}\rbra{\E\sbra{\operatorname{Poi}\rbra{U}}, \E\sbra{\operatorname{Poi}\rbra{V}}} \leq \rbra*{\frac{2eM}{L}}^L.
    \]
\end{lemma}

\begin{lemma}[Inequalities for Poisson distributions] \label{lemma:poi-neq}
    For independent $X, Y \sim \operatorname{Poi}\rbra{\lambda}$, integer $q \geq 1$, real $u > 0$, and $0 < v < 1$, 
    \begin{align}
        \E\sbra{X^q} & \leq q^q \max\cbra{\lambda, \lambda^q}, \label{eq:EXq-Poi} \\
        \E\sbra{\abs{X - \lambda}^q} & \leq \rbra{2q}^q \lambda^{q/2} \quad \text{for } \lambda \geq 1, \label{eq:EX-lambda-q} \\
        \E\sbra{\rbra{1+X+Y+\lambda}^q} & \leq \rbra{4q}^q \rbra{1+\lambda}^q, \label{eq:XY-Poi} \\
        \Pr\sbra{X \geq \rbra{1+u}\lambda} & \leq \rbra*{\frac{e^u}{\rbra{1+u}^{1+u}}}^{\lambda}, \label{eq:Poi-upp} \\
        \Pr\sbra{X \leq \rbra{1-v}\lambda} & \leq e^{-v^2\lambda/2}. \label{eq:Poi-low}
    \end{align}
\end{lemma}
\begin{proof}
    \cref{eq:EXq-Poi} can be found in \cite[Corollary 2.1]{BT10} (cf.\ \cite[Equation (2)]{Ahl22}). 
    \cref{eq:EX-lambda-q} can be derived by \cite[Lemma 3.1]{Pri11}. 
    \cref{eq:XY-Poi} can be obtained from \cref{eq:EXq-Poi} with the Minkowski inequality. 
    \cref{eq:Poi-upp,eq:Poi-low} can be found in \cite[Lemma 21]{JVHW15} (see also \cite[Exercise 4.7]{MU05}).
\end{proof}

\section{Lower Bounds}

In this section, we show matching lower bounds for R\'enyi entropy estimation of noninteger order. 
To this end, we first present a constructive moment matching in \cref{sec:mm} for our purpose, and then use it to show matching lower bounds in \cref{subsec:lb}. 

\subsection{Constructive Moment Matching} \label{sec:mm}

\begin{lemma} \label{lemma:mm}
    Fix a noninteger $\alpha > 0$ and a real $R > 1$. 
    There is a constant $c_{\alpha, R} > 0$ such that for every integer $K \geq 2$, there is a $K$-moment matching $\rbra{\nu_0, \nu_1}$ on $\sbra{0, K^2}$ with mean $\mu$ such that
    \begin{align}
        c_{\alpha, R} \leq \mu, \mu_b^{\rbra{2}}, \mu_b^{\rbra{\alpha}}, \mu_b^{\rbra{2\alpha}} \leq 2, \label{eq:cond1} \\
        \mu_1^{\rbra{\alpha}} \geq R \mu_0^{\rbra{\alpha}}, \label{eq:cond2}
    \end{align}
    where $\mu_b^{\rbra{\alpha}} = \E_{X \sim \nu_b} \sbra{X^\alpha}$ for $b \in \cbra{0, 1}$. 
\end{lemma}

\begin{proof}
    Take $c_{\alpha, R} = \eta^{1+ 2r \max\cbra{\alpha, 1}} / 12$, where $r = \ceil{2\alpha} + 3 \geq 4$ and $\eta = \rbra{2^{r+3} R}^{-1/\min\cbra{\alpha - \floor{\alpha}, \ceil{\alpha} - \alpha}} < 1/4$.
    For every $K \geq 2$, we construct a $K$-moment matching $\rbra{\nu_0, \nu_1}$ supported on the set
    \[
    \mathcal{X} = \cbra[\Big]{0, \underbrace{\eta^r, \eta^{r-1}, \dots, \eta}_{r \textup{ elements}}, \underbrace{1^2, 2^2, \dots, K^2}_{K \textup{ elements}}} \subseteq \sbra{0, K^2}
    \]
    such that for each $x \in \mathcal{X}$ and $b \in \cbra{0, 1}$, 
    \[
    \nu_b\rbra{x} = \frac{1}{Z} \max\cbra*{\rbra{-1}^{K+r-\ceil{\alpha}+1-b} w\rbra{x}, 0}, \qquad 
    w\rbra{x} = \prod_{z \in \mathcal{X} \setminus \cbra{x}} \frac{1}{x - z}, \qquad Z = \frac{1}{2} \sum_{x \in \mathcal{X}} \abs{w\rbra{x}}.
    \]
    Lagrange interpolation gives that 
    \[
    \sum_{x \in \mathcal{X}} w\rbra{x} x^j = 0, \qquad 0 \leq j \leq \abs{\mathcal{X}} - 2 = K+r-1,
    \]
    which ensures that $\rbra{\nu_0, \nu_1}$ is a $K$-moment matching. 
    Moreover, direct calculation gives 
    \begin{align}
        \abs{w\rbra{0}} & = \frac{1}{\eta^{r\rbra{r+1}/2} \rbra{K!}^2}, \label{eq:w0} \\
        \abs{w\rbra{j^2}} & = \frac{2}{\rbra{K-j}! \rbra{K+j}! \prod_{\ell=1}^r \rbra{j^2 - \eta^{r+1-\ell}}} \leq \frac{2^{r+1}}{\rbra{K!}^2 j^{2r}}, & 1 \leq j \leq K, \label{eq:wj2} \\
        \frac{\abs{w\rbra{\eta^{r+1-\ell}}}}{\abs{w\rbra{0}}} & = \eta^{\ell\rbra{\ell-1}/2} W_\ell, \qquad W_\ell = \prod_{\substack{1 \leq h \leq r \\ h \neq \ell}} \frac{1}{1 - \eta^{\abs{h - \ell}}} \prod_{j=1}^K \frac{1}{1 - \eta^{r+1-\ell}/j^2}, & 1 \leq \ell \leq r. \nonumber
    \end{align}
    Since $r \geq 4$ and $0 < \eta < 1/4$, we have 
    \begin{align*}
        1 & \leq W_\ell \leq \exp\rbra*{\frac{4\eta}{1-\eta} + 4\eta} < 16, \\
        1 & \leq \frac{Z}{\abs{w\rbra{0}}} \leq \frac{1}{2} \rbra*{1 + \frac{16}{1-\eta} + 2^{r+2} \eta^{r\rbra{r+1}/2}} < 12.
    \end{align*}

    We first show \cref{eq:cond1}. 
    For every $0 < q \leq \max\cbra{2, 2\alpha}$, 
    \begin{align*}
    \mu_b^{\rbra{q}} 
    & = \E_{X \sim \nu_b} \sbra{X^q} = \sum_{\ell=1}^r \nu_b\rbra{\eta^{r+1-\ell}} \eta^{q\rbra{r+1-\ell}} + \sum_{j=1}^K \nu_b\rbra{j^2} j^{2q} \\
    & \leq \sum_{\ell=1}^r \nu_b\rbra{\eta^{r+1-\ell}} + \frac{1}{Z} \sum_{j=1}^K \abs{w\rbra{j^2}} j^{2q} \leq 1 + 1 = 2,
    \end{align*}
    where we use the fact by \cref{eq:w0,eq:wj2} that
    \[
    \frac{1}{Z} \sum_{j=1}^K \abs{w\rbra{j^2}} j^{2q} \leq 2^{r+1} \eta^{r\rbra{r+1}/2} \sum_{j=1}^K j^{-2\rbra{r-q}} \leq 2^{r+2} \eta^{r\rbra{r+1}/2} \leq 1.
    \]
    On the other hand, 
    \[
    \mu_b^{\rbra{q}} \geq \eta^{rq} \nu_b\rbra{\sbra{\eta^r, 1}} \geq \frac{\eta^{rq+1}}{12},
    \]
    where we use the fact that $w\rbra{\eta^r} w\rbra{\eta^{r-1}} < 0$ and thus
    \[
    \nu_b\rbra{\sbra{\eta^r, 1}} \geq \nu_b\rbra{\eta^r} + \nu_b\rbra{\eta^{r-1}} \geq \min\cbra*{\frac{\abs{w\rbra{\eta^r}}}{Z}, 
    \frac{\abs{w\rbra{\eta^{r-1}}}}{Z}} \geq \frac{1}{12} \min\cbra*{\frac{\abs{w\rbra{\eta^r}}}{\abs{w\rbra{0}}}, 
    \frac{\abs{w\rbra{\eta^{r-1}}}}{\abs{w\rbra{0}}}} \geq \frac{\eta}{12}. 
    \]

    Then, we show \cref{eq:cond2}. 
    Let $x_* = \eta^{r+1-\ceil{\alpha}}$.
    Note that $\nu_0\rbra{x_*} = 0$ and $\nu_1\rbra{x_*} = \abs{w\rbra{x_*}}/Z > 0$.
    Then, 
    \begin{align*}
        \frac{\mu_0^{\rbra{\alpha}}}{\mu_1^{\rbra{\alpha}}} 
        & = \frac{\sum_{x \in \mathcal{X} \setminus \cbra{0, x_*}} \nu_0\rbra{x} x^\alpha}{ \nu_1\rbra{x_*} x_*^\alpha + \sum_{x \in \mathcal{X} \setminus \cbra{0, x_*}} \nu_1\rbra{x} x^\alpha} 
        \leq \frac{\sum_{x \in \mathcal{X} \setminus \cbra{0, x_*}} \abs{w\rbra{x}} x^\alpha}{ \abs{w\rbra{x_*}} x_*^\alpha } \\
        & = \frac{1}{\abs{w\rbra{x_*}} x_*^\alpha} \rbra*{\sum_{\substack{1 \leq \ell \leq r \\ \ell \neq \ceil{\alpha}}} \abs{w\rbra{\eta^{r+1-\ell}}} \eta^{\alpha\rbra{r+1-\ell}} + \sum_{j=1}^K \abs{w\rbra{j^2}} j^{2\alpha} }.
    \end{align*}
    Dividing both numerator and denominator by $\abs{w\rbra{0}}$ and using the fact that $1 \leq W_\ell < 16$ gives
    \begin{align*}
        \frac{\mu_0^{\rbra{\alpha}}}{\mu_1^{\rbra{\alpha}}}
        & \leq \sum_{\substack{1 \leq \ell \leq r \\ \ell \neq \ceil{\alpha}}}
        \frac{W_\ell}{W_{\ceil{\alpha}}}
        \eta^{\rbra{\ell-\ceil{\alpha}}\rbra{\frac{\ell+\ceil{\alpha}-1}{2}-\alpha}} + \frac{2^{r+1}}{W_{\ceil{\alpha}}}
        \eta^{\rbra{r+1-\ceil{\alpha}}\rbra{\frac{r+\ceil{\alpha}}{2}-\alpha}}
        \sum_{j=1}^K j^{-2\rbra{r-\alpha}} \\
        & \leq \rbra*{16\rbra{r-1}+2^{r+2}}
        \eta^{\min\cbra{\alpha-\floor{\alpha},\ceil{\alpha}-\alpha}} \leq \frac{1}{R},
    \end{align*}
    which gives $\mu_1^{\rbra{\alpha}} \geq R \mu_0^{\rbra{\alpha}}$. 
\end{proof}

\subsection{The lower bound} \label{subsec:lb}

\begin{theorem} \label{thm:lower-bound}
    Fix noninteger $\alpha > 0$. 
    For any sufficiently large $d \geq 2$ and any sufficiently small $0 < \varepsilon \leq 1/16$, estimating the $\alpha$-R\'enyi entropy of an unknown $d$-dimensional discrete distribution to within additive error $\varepsilon$ requires sample complexity $\Omega\rbra{\frac{d^{\max\cbra{1/\alpha, 1}}}{\varepsilon^{1/\alpha}\log\rbra{d}}}$. 

    Combining the lower bounds in \cref{lemma:lower-bound-partial}, we yield the sample complexity lower bound 
    \[
    \Omega\rbra*{\frac{d^{\max\cbra{1/\alpha, 1}}}{\varepsilon^{1/\alpha}\log\rbra{d}} + \frac{d^{\abs{1-1/\alpha}}}{\varepsilon^2}}.
    \]
\end{theorem}

\begin{proof}
    When $\alpha > 1$ and $\frac{d^{1-1/\alpha}}{\varepsilon^2} \geq \frac{d}{\varepsilon^{1/\alpha} \log\rbra{d}}$, the desired lower bound is already given by \cref{lemma:lower-bound-partial}. 
    In the rest of the proof, we only have to consider the case where $0 < \alpha < 1$ and the case where $\alpha > 1$ and $\frac{d^{1-1/\alpha}}{\varepsilon^2} < \frac{d}{\varepsilon^{1/\alpha} \log\rbra{d}}$.
    
    Let $\rbra{\nu_0, \nu_1}$ be the $K$-moment matching on $\sbra{0, K^2}$ with mean $\mu$ given in \cref{lemma:mm} with 
    \[
    K \coloneqq \ceil{4 \log\rbra{ed}}, \qquad R \coloneqq \rbra{2+4^\alpha} e^{\abs{1-\alpha}/2}, 
    \]
    such that $c_{\alpha, R} \leq \mu, \mu_b^{\rbra{2}}, \mu_b^{\rbra{\alpha}}, \mu_b^{\rbra{2\alpha}} \leq 2$ and $\mu_1^{\rbra{\alpha}} \geq R \mu_0^{\rbra{\alpha}}$, 
    where $\mu_b^{\rbra{q}} = \E_{X \sim \nu_b} \sbra{X^q}$ for $b \in \cbra{0, 1}$. 
    
    Now we construct the hard instance.
    Let
    \begin{align*}
    \rbra{m, s} & = \begin{cases}
        \rbra{d-1, 1}, & 0 < \alpha < 1, \\
        \rbra{\floor{d/2}, \floor{d/2}}, & \alpha > 1,
    \end{cases} \\
    t & = \rbra*{ m\mu + m^{1/\alpha}s^{1-1/\alpha}\rbra*{ \frac{\mu_1^{\rbra{\alpha}} - e^{8\abs{1-\alpha}\varepsilon}  \mu_0^{\rbra{\alpha}}}{e^{8\abs{1-\alpha}\varepsilon} - 1} }^{1/\alpha} }^{-1} = \Theta\rbra*{\frac{\varepsilon^{1/\alpha}}{m^{1/\alpha} s^{1-1/\alpha}}} = \Theta\rbra*{\frac{\varepsilon^{1/\alpha}}{d^{\max\cbra{1/\alpha, 1}}}}.
    \end{align*}
    For $b \in \cbra{0, 1}$, let $U_1, U_2, \dots, U_m \sim \nu_b$ be independent and identical random variables. 
    Denote $R_U = 1 - t\sum_{i=1}^m U_i$. 
    Define the probability distribution:
    \[
    p\rbra{U} = \begin{cases}
        \rbra{ \underbrace{tU_1, tU_2, \dots, tU_m}_{m}, \underbrace{R_U/s, \dots, R_U/s}_{s}, \underbrace{0, \dots, 0}_{d-m-s} }, & R_U \geq 1/2, \\
        \rbra{ 1, \underbrace{0, \dots, 0}_{d-1} }, & R_U < 1/2.
    \end{cases}
    \]
    This hard instance is inspired by the one for quantum estimation problems in \cite{Wan26}, which generalizes the one used in \cite{WY16} where $s = 1$.
    Consider the hypothesis testing problem: determine a hypothesis $b \in \cbra{0, 1}$, each case with equal probability, given (Poissonized) sampling access to the unknown distribution $p\rbra{U}$. 
    Then, for sufficiently large $d \geq 2$, there exist $h_0, h_1 \in \mathbb{R}$ with $\abs{h_0-h_1} = 8\varepsilon$ such that for every $b \in \cbra{0, 1}$, 
    \begin{enumerate}[label={\upshape(\arabic*)}, topsep=0.33em, itemsep=0.33em, parsep=0.33em]
        \item $\Pr\sbra{ \abs{\mathrm{H}_\alpha\rbra{p\rbra{U}} - h_b} > \varepsilon \,|\, b } \leq 0.01$. \label{cond1}
        \item $d_{\textup{TV}}\rbra{P_0^{\rbra{n}}, P_1^{\rbra{n}}} < 0.1$ for all $0 \leq n \leq c_{\alpha,R}/\rbra{2^{16}tK}$, where $P_b^{\rbra{n}} = \E\sbra{ \bigotimes_{i=1}^d\operatorname{Poi}\rbra{np_i\rbra{U}} \,|\, b}$. \label{cond2}
    \end{enumerate}
    Condition \ref{cond1} means that an $\alpha$-R\'enyi entropy estimator within additive error $\varepsilon$ can be used to determine $b$ with success probability at least $2/3 - 0.01$, whereas Condition \ref{cond2} means that $n$ Poissonized samples will determine $b$ with success probability at most $0.55$, giving a contradiction. 
    Therefore, any $\alpha$-R\'enyi entropy estimator within additive error $\varepsilon$ requires sample complexity at least
    \[
    \frac{c_{\alpha, R}}{2^{16} t K} \geq \Omega\rbra*{\frac{d^{\max\cbra{1/\alpha, 1}}}{\varepsilon^{1/\alpha}\log\rbra{d}}}.
    \]

    To complete the proof, it remains to prove Conditions \ref{cond1} and \ref{cond2}.
    For Condition \ref{cond1}, we take
    \[
    h_b = \frac{1}{1-\alpha} \log\rbra*{f_b}, \qquad f_b = mt^{\alpha} \mu_b^{\rbra{\alpha}} + s^{1-\alpha}\rbra{1-mt\mu}^\alpha .
    \]
    It can be verified that $f_1 = e^{8\abs{1-\alpha}\varepsilon} f_0$ and thus $\abs{h_0 - h_1} = 8\varepsilon$. 
    The probability $\Pr\sbra{ \abs{\mathrm{H}_\alpha\rbra{p\rbra{U}} - h_b} > \varepsilon \,|\, b }$ can be split into two cases (whether $R_U \geq 1/2$ or not):
    \begin{align} \label{eq:prob-split}
        \Pr\sbra*{ \abs{\mathrm{H}_\alpha\rbra{p\rbra{U}} - h_b} > \varepsilon \mid b }
        & \leq \Pr\sbra*{ R_U < 1/2 \mid b } + \Pr\sbra*{ R_U \geq 1/2, \abs{\mathrm{H}_\alpha\rbra{p\rbra{U}} - h_b} > \varepsilon \mid b }.
    \end{align}
    For the first term of \cref{eq:prob-split}, Chebyshev's inequality gives 
    \begin{align}
    \Pr\sbra*{ R_U < 1/2 \mid b } 
    & \leq \Pr\sbra*{ \abs*{R_U - \E\sbra*{R_U}} \geq \E\sbra*{R_U} - 1/2 \mid b } \nonumber \\
    & \leq \frac{\Var\sbra*{R_U \mid b}}{\rbra{\E\sbra*{R_U} - 1/2}^2} \leq 24mt^2, \label{eq:prob-RU<1/2}
    \end{align}
    where we use the fact that $\Var\sbra*{R_U | b} = t^2 \sum_{i=1}^m \Var\sbra{U_i | b} \leq m t^2 \mu_b^{\rbra{2}} \leq 2mt^2$ and $\E\sbra{R_U} = 1 - mt\mu \geq 4/5$.
    To see that $1 - mt\mu \geq 4/5$, according to the definition of $t$, we have
    \[
    t \leq \frac{W\varepsilon^{1/\alpha}}{m^{1/\alpha}s^{1-1/\alpha}}, \qquad W = \max\cbra*{\rbra*{\frac{8\abs{1-\alpha}e^{\abs{1-\alpha}/2}}{c_{\alpha, R}}}^{1/\alpha}, 1}.
    \]
    Then, for $0 < \alpha < 1$, when $m \geq \rbra{10W}^{\alpha/\rbra{1-\alpha}}$ for sufficiently large $d$, we have $mt\mu \leq 2W \varepsilon^{1/\alpha} m^{1-1/\alpha} \leq 1/5$; for $\alpha > 1$, using $m = s$, $\mu^{\alpha} \leq \mu_0^{\rbra{\alpha}}$, and $\mu_1^{(\alpha)}-e^{8\abs{1-\alpha}\varepsilon}\mu_0^{(\alpha)} \geq \rbra{R-e^{8\abs{1-\alpha}\varepsilon}}\mu_0^{(\alpha)}$, we have 
    \[
    \frac{mt\mu}{1-mt\mu} \leq \rbra*{ \frac{e^{8\abs{1-\alpha}\varepsilon} - 1}{R - e^{8\abs{1-\alpha}\varepsilon}} }^{1/\alpha} \leq \frac{1}{4}.
    \]
    
    For the second term of \cref{eq:prob-split}, noting that $\abs{\mathrm{F}_\alpha\rbra{p\rbra{U}} - f_b} \leq \theta_\alpha \varepsilon f_b$ with $\theta_\alpha = \min\cbra{\abs{1-\alpha}, 1}/4 \leq 1/4$ and $\varepsilon \leq 1/16$ implies $\abs{\mathrm{H}_\alpha\rbra{p\rbra{U}} - h_b} \leq \varepsilon$, we have
    \begin{align}
        \Pr\sbra*{ R_U \geq 1/2, \abs{\mathrm{H}_\alpha\rbra{p\rbra{U}} - h_b} > \varepsilon \mid b }
        & \leq \Pr\sbra*{ R_U \geq 1/2, \abs{\mathrm{F}_\alpha\rbra{p\rbra{U}} - f_b} > \theta_\alpha \varepsilon f_b \mid b } \nonumber \\
        & \leq \frac{ \E\sbra{\mathbf{1}_{\cbra{R_U \geq 1/2}} \abs{\mathrm{F}_\alpha\rbra{p\rbra{U}} - f_b}^2 | b} }{ \theta_\alpha^2 \varepsilon^2 f_b^2 }, \label{eq:prob-RU>1/2}
    \end{align}
    where we use Markov's inequality.
    Note that when $R_U \geq 1/2$,
    \begin{align*}
        \mathrm{F}_\alpha\rbra{p\rbra{U}} - f_b = t^{\alpha} \sum_{i=1}^m \rbra*{U_i^\alpha-\mu_b^{\rbra{\alpha}}} + s^{1-\alpha} \rbra*{R_U^\alpha-\rbra{\E\sbra{R_U}}^\alpha}.
    \end{align*}
    When $R_U \geq 1/2$, it holds that $R_U, \E\sbra{R_U} \in \sbra{1/2, 1}$, and thus Lagrange's mean value theorem gives 
    \[
    \abs*{R_U^\alpha - \rbra{\E\sbra{R_U}}^\alpha} \leq \sup_{x \in \sbra{1/2, 1}} \alpha x^{\alpha-1} \abs{R_U - \E\sbra{R_U}} = \alpha 2^{\max\cbra{1 - \alpha, 0}} \abs{R_U - \E\sbra{R_U}}.
    \]
    Using $R_U - \E\sbra{R_U} = -t \sum_{i=1}^m \rbra{U_i - \mu}$, we have
    \begin{align*}
        \abs*{\mathrm{F}_\alpha\rbra{p\rbra{U}} - f_b} \leq t^\alpha \abs*{ \sum_{i=1}^m \rbra*{U_i^\alpha-\mu_b^{\rbra{\alpha}}} } + \alpha 2^{\max\cbra{1 - \alpha, 0}} s^{1-\alpha} t \abs*{\sum_{i=1}^m \rbra{U_i - \mu}}.
    \end{align*}
    Using $\rbra{a + b}^2 \leq 2a^2 + 2b^2$ and $\mathbf{1}_{\cbra{R_U \geq 1/2}} \leq 1$, direct calculation gives 
    \begin{align}
         \E\sbra{\mathbf{1}_{\cbra{R_U \geq 1/2}} \abs{\mathrm{F}_\alpha\rbra{p\rbra{U}} - f_b}^2 | b} 
         & \leq 2t^{2\alpha} \E\sbra*{ \rbra*{\sum_{i=1}^m \rbra*{U_i^\alpha-\mu_b^{\rbra{\alpha}}}}^2 \mid b } \nonumber \\
         & \qquad + 2\alpha^2 4^{\max\cbra{1-\alpha, 0}} s^{2-2\alpha} t^2 \E\sbra*{ \rbra*{\sum_{i=1}^m \rbra{U_i - \mu}}^2 \mid b} \nonumber \\
         & \leq 2t^{2\alpha} m \mu_b^{\rbra{2\alpha}} + 2\alpha^2 4^{\max\cbra{1-\alpha, 0}} s^{2-2\alpha} t^2 m \mu_b^{\rbra{2}} \nonumber \\
         & \leq 4mt^{2\alpha} + 4\alpha^2 4^{\max\cbra{1-\alpha, 0}} s^{2-2\alpha} m t^2, \label{eq:expand}
    \end{align}
    where we use $\mu_b^{\rbra{2}}, \mu_b^{\rbra{2\alpha}} \leq 2$. 
    Combining \cref{eq:prob-split,eq:prob-RU<1/2,eq:prob-RU>1/2,eq:expand} with $f_b \geq \rbra{4/5}^\alpha s^{1-\alpha}$ gives
    \begin{align*}
        \Pr\sbra*{ \abs{\mathrm{H}_\alpha\rbra{p\rbra{U}} - h_b} > \varepsilon \mid b }
        & \leq 24mt^2 + \frac{4\rbra{5/4}^{2\alpha}}{\theta_\alpha^2} \rbra*{ \frac{mt^{2\alpha}}{s^{2-2\alpha} \varepsilon^2} + \frac{\alpha^2 4^{\max\cbra{1-\alpha, 0}} m t^2}{\varepsilon^2}} \\
        & \leq \frac{M}{m \varepsilon^{\max\cbra{2-2/\alpha, 0}}} \leq \frac{1}{100},
    \end{align*}
    where $M$ is a constant only depending on $\alpha$ given by
    \begin{align*}
    M 
    = 24W^2 + \frac{4\rbra{5/4}^{2\alpha}}{\theta_\alpha^2} \rbra*{ W^{2\alpha} + \alpha^2 4^{\max\cbra{1-\alpha, 0}} W^2 }. 
    \end{align*}
    To complete the proof, it remains to show that $100 M \leq m\varepsilon^{\max\cbra{2-2/\alpha, 0}}$ for sufficiently large $d$.
    For $0 < \alpha < 1$, this is obtained by noting that $m \geq 100M$ for sufficiently large $d$. 
    For $\alpha > 1$, if $\frac{d}{\varepsilon^{1/\alpha} \log\rbra{d}} \geq \frac{d^{1-1/\alpha}}{\varepsilon^2}$, i.e., $d\varepsilon^{2\alpha-1} \geq \rbra{\log\rbra{d}}^\alpha$, then $m \varepsilon^{2-2/\alpha} \geq d \varepsilon^{2\alpha-1}/3 \geq \rbra{\log\rbra{d}}^\alpha/3 \geq 100 M$ for $d \geq \exp\rbra{\rbra{300M}^{1/\alpha}}$; otherwise, the lower bound $\Omega\rbra{\frac{d}{\varepsilon^{1/\alpha} \log\rbra{d}}}$ obtained here is covered by the lower bound $\Omega\rbra{\frac{d^{1-1/\alpha}}{\varepsilon^2}}$ given in \cref{lemma:lower-bound-partial}.

    For Condition \ref{cond2}, define the distributions
    \[
    Q_b^{\rbra{n}} = \rbra*{\E_{U \sim \nu_b}\sbra*{\operatorname{Poi}\rbra{ntU}}}^{\otimes m} \otimes \operatorname{Poi}\rbra{n\rbra{1-mt\mu}/s}^{\otimes s} \otimes \operatorname{Poi}\rbra{0}^{\otimes \rbra{d-m-s}}.
    \]
    Triangle inequality gives 
    \begin{align} \label{eq:dTV-P01}
    d_{\textup{TV}}\rbra{P_0^{\rbra{n}}, P_1^{\rbra{n}}} \leq d_{\textup{TV}}\rbra{P_0^{\rbra{n}}, Q_0^{\rbra{n}}} + d_{\textup{TV}}\rbra{Q_0^{\rbra{n}}, Q_1^{\rbra{n}}} + d_{\textup{TV}}\rbra{Q_1^{\rbra{n}}, P_1^{\rbra{n}}}.
    \end{align}
    For the first and the third terms, we have
    \begin{align}
        d_{\textup{TV}}\rbra{P_b^{\rbra{n}}, Q_b^{\rbra{n}}} \leq \Pr\sbra{R_U < 1/2 | b} + \E\sbra*{\mathbf{1}_{\cbra{R_U \geq 1/2}} d_{\textup{TV}}\rbra{\operatorname{Poi}\rbra{nR_U/s}^{\otimes s}, \operatorname{Poi}\rbra{n\rbra{1-mt\mu}/s}^{\otimes s}} }. \label{eq:dTV-PQ}
    \end{align}
    When $R_U \geq 1/2$, 
    \begin{align}
    d_{\textup{TV}}\rbra{\operatorname{Poi}\rbra{nR_U/s}^{\otimes s}, \operatorname{Poi}\rbra{n\rbra{1-mt\mu}/s}^{\otimes s}}
    & \leq \sqrt{\frac 1 2 D_{\textup{KL}}\rbra{\operatorname{Poi}\rbra{nR_U/s}^{\otimes s}, \operatorname{Poi}\rbra{n\rbra{1-mt\mu}/s}^{\otimes s}}} \nonumber \\
    & = \sqrt{\frac 1 2 n \rbra*{R_U \log\rbra*{\frac{R_U}{1-mt\mu}} - R_U + \rbra{1-mt\mu}} } \nonumber \\
    & \leq \sqrt{\frac{n}{2\rbra{1-mt\mu}}} \abs*{R_U - \rbra{1-mt\mu}}, \label{eq:dTV-Poi}
    \end{align}
    where we use the Pinsker inequality $d_{\textup{TV}}\rbra{p, q} \leq \sqrt{D_{\textup{KL}}\rbra{p, q}/2}$ and the fact that $D_{\textup{KL}}\rbra{\operatorname{Poi}\rbra{a} \| \operatorname{Poi}\rbra{b}} = a \log\rbra{a/b} - a + b$. 
    Combining \cref{eq:dTV-PQ,eq:dTV-Poi,eq:prob-RU<1/2} gives 
    \begin{align}
        d_{\textup{TV}}\rbra{P_b^{\rbra{n}}, Q_b^{\rbra{n}}} 
        & \leq 24mt^2 + \sqrt{\frac{n}{2\rbra{1-mt\mu}} } \E\sbra*{ \mathbf{1}_{\cbra{R_U \geq 1/2}} \abs*{R_U - \rbra{1-mt\mu}} \mid b} \nonumber\\
        & \leq 24mt^2 + \sqrt{\frac{n}{2\rbra{1-mt\mu}}} \sqrt{mt^2 \mu_b^{\rbra{2}}} \nonumber\\
        & \leq 24mt^2 + 2 \sqrt{nmt^2}, \label{eq:dTV-PQ-up}
    \end{align}
    where we use the fact that $\mu_b^{\rbra{2}} \leq 2$ and $1-mt\mu \geq 4/5$. 
    For the second term, verifying that $2entK \leq 2ec_{\alpha, R} / 2^{16} < 1/8$ which gives $K > 2e\rbra{ntK^2}$, \cref{lemma:WY16-lemma3} gives 
    \begin{align} \label{eq:dTV-Q01}
    d_{\textup{TV}}\rbra{Q_0^{\rbra{n}}, Q_1^{\rbra{n}}} \leq m \rbra*{2entK}^K.
    \end{align}
    Combining \cref{eq:dTV-P01,eq:dTV-PQ-up,eq:dTV-Q01} gives that for any $0 \leq n \leq c_{\alpha,R}/\rbra{2^{16}tK}$,
    \begin{align*}
        d_{\textup{TV}}\rbra{P_0^{\rbra{n}}, P_1^{\rbra{n}}}
        & \leq 48mt^2 + 4 \sqrt{nmt^2} + m \rbra*{2entK}^K \\
        & \leq \frac{48}{2^{12}} + \frac{4}{\sqrt{5 \cdot 2^{16}}} + e^{-\rbra{4\log\rbra{8} - 1}} < 0.1 
    \end{align*}
    for sufficiently large $d$ such that $m \geq 2^{12} W^2$, where we use $mt^2 \leq W^2/m$, $nmt^2 \leq 1/\rbra{5\cdot 2^{16}K}$, and $2entK < 1/8$. 
\end{proof}

\section{Upper Bounds}

In this section, we present our estimator, with the analysis of its correctness and sample complexity. 

\subsection{The estimator}

For an unknown discrete distribution $p = \rbra{p_1, p_2, \dots, p_d}$, draw (expected) $2n$ samples in the Poissonized sampling model; i.e., let independent $X_i, Y_i \sim \operatorname{Poi}\rbra{np_i}$ for $1 \leq i \leq d$. 
Let 
\[
J = \ceil{4\alpha} + 3, \qquad K = \max\cbra*{\floor*{\frac{L}{32 \log \rbra{2}}}, 1}, \qquad L = \log\rbra{ed}.
\]
Define 
\begin{equation}
\widehat{H} = \frac{1}{1-\alpha} \log\rbra*{\operatorname{clip}_{I_\alpha}\rbra{\widehat{F}}}, \qquad
\widehat{F} = \sum_{i=1}^d \rbra*{ \mathbf{1}_{\cbra{Y_i \leq 10L}} \mathsf{L}_n\rbra{X_i} + \mathbf{1}_{\cbra{Y_i > 10L}} \mathsf{H}_n\rbra{X_i, Y_i} }, \label{eq:def-estimator}
\end{equation}
where 
\begin{align}
    I_\alpha & = \begin{cases}
        \sbra{1, d^{1-\alpha}}, & 0 < \alpha < 1, \\
        \sbra{d^{1-\alpha}, 1}, & \alpha > 1,
    \end{cases} \nonumber \\
    \mathsf{L}_n\rbra{x} & = \rbra*{\frac{40L}{n}}^{\alpha} \mathbf{1}_{\cbra{x \leq 80L}} \sum_{j=0}^K \frac{a_j}{\rbra{40L}^j} x^{\underline{j}}, \label{eq:def-Ln} \\
    \mathsf{H}_n\rbra{x, y} & = \frac{1}{n^\alpha} \sum_{j=0}^J \binom{\alpha}{j} \sum_{r=0}^j \binom{j}{r} \rbra{-1}^{j-r} y^{\alpha-r} x^{\underline{r}}, \label{eq:def-Hn}
\end{align}
where $a_j$'s are the polynomial coefficients of $Q_K\rbra{x} = \sum_{j=0}^K a_j x^j$ given in \cref{lemma:poly-approx-xalpha}, and we use the convention that $0^x = 0$ for any $x \neq 0$ and $0^0 = 1$. 

We have the following result for the estimator $\widehat{H}$. 

\begin{theorem} \label{thm:estimator}
    By taking 
    \[
    n = \Theta\rbra*{ \frac{d^{\max\cbra{1, 1/\alpha}}}{\varepsilon^{1/\alpha} L} + \mathbf{1}_{\cbra{\alpha>1/2}} \frac{d^{\abs{1-1/\alpha}}}{\varepsilon^2} },
    \]
    the estimator $\widehat{H}$ defined in \cref{eq:def-estimator} satisfies $\Pr\sbra{ \abs{\widehat{H} - \mathrm{H}_\alpha\rbra{p}} \leq \varepsilon } \geq 2/3$ for all distributions $p$ over an alphabet of size $d$. 
\end{theorem}

\subsection{Correctness and complexity analysis}

\begin{lemma} \label{lemma:Poi-single-error}
    For independent $X, Y \sim \operatorname{Poi}\rbra{np}$, let 
    \[
    Z = \mathbf{1}_{\cbra{Y \leq 10L}} \mathsf{L}_n\rbra{X} + \mathbf{1}_{\cbra{Y > 10L}} \mathsf{H}_n\rbra{X, Y},
    \]
    where $\mathsf{L}_n$ and $\mathsf{H}_n$ are defined in \cref{eq:def-Ln,eq:def-Hn}, respectively. 
    Then, 
    \begin{align}
        \abs*{ \E\sbra{Z} - p^{\alpha} } & \leq C_\alpha \rbra{nL}^{-\alpha}, \label{eq:cond1-abs-bias} \\
        \Var\sbra{Z} & \leq C_\alpha^2 \rbra*{ \frac{L^{2\alpha} d^{1/4}}{n^{2\alpha}} + \mathbf{1}_{\cbra{\alpha > 1/2}} \frac{p^{2\alpha - 1}}{n} }, \label{eq:cond2-var}
    \end{align}
    where $C_\alpha = c_\alpha \rbra{256 J}^{32 J}$ and $c_\alpha$ is the constant in \cref{lemma:poly-approx-xalpha}.
\end{lemma}

\begin{proof}
    Write $\lambda = np$. 

    \textbf{Case 1: $Y \leq 10L$.} 
    We consider the term $\mathbf{1}_{\cbra{Y \leq 10L}} \mathsf{L}_n\rbra{X}$. 

    \textbf{Case 1.1: $\lambda \leq 20L$.}
    Let $T\rbra{x} = \rbra{40L}^\alpha \sum_{j=0}^K a_j x^{\underline{j}}/\rbra{40L}^{j}$. 
    Then,
    \[
    \E\sbra{T\rbra{X}} = \rbra{40L}^{\alpha} Q_K\rbra{\lambda/40L}, \qquad
    n^\alpha \mathsf{L}_n\rbra{X} = T\rbra{X} - \mathbf{1}_{\cbra{X > 80L}}T\rbra{X}.
    \]
    The triangle inequality gives 
    \begin{equation} \label{eq:case1.1-abs-bias}
    \abs*{ \E\sbra{n^\alpha \mathsf{L}_n\rbra{X}} - \lambda^\alpha } \leq \abs*{ \E\sbra{T\rbra{X}} - \lambda^\alpha } + \abs*{ \E\sbra{\mathbf{1}_{\cbra{X > 80L}}T\rbra{X}} }.
    \end{equation}
    For the first term, \cref{lemma:poly-approx-xalpha} gives
    \begin{equation} \label{eq:case1.1-abs-diff-TX}
    \abs*{ \E\sbra{T\rbra{X}} - \lambda^\alpha } = \rbra{40L}^{\alpha} \abs*{ Q_K\rbra*{\frac{\lambda}{40L}} - \rbra*{\frac{\lambda}{40L}}^\alpha } \leq c_\alpha \rbra{40L}^{\alpha} K^{-2\alpha}.
    \end{equation}
    For the second term, 
    \begin{align}
        \abs*{ \E\sbra{\mathbf{1}_{\cbra{X > 80L}}T\rbra{X}} }
        & \leq \rbra{40L}^\alpha \sum_{j=0}^K \frac{\abs{a_j}}{\rbra{40L}^j} \E\sbra{\mathbf{1}_{\cbra{X > 80L}} X^{\underline{j}}} \nonumber \\
        & \leq \rbra{40L}^\alpha \sum_{j=0}^K \frac{\abs{a_j}}{\rbra{40L}^j} \lambda^j e^{-20L} \nonumber \\
        & \leq \rbra{40L}^\alpha e^{-20L} \sum_{j=0}^K \abs{a_j} \nonumber \\
        & \leq c_\alpha 8^K \rbra{40L}^\alpha e^{-20L}, \label{eq:case1.1-abs-TX-large}
    \end{align}
    where the second inequality uses $\E\sbra{\mathbf{1}_{\cbra{X > 80L}} X^{\underline{j}}}
    \leq 2^{-80L} \E\sbra{2^X X^{\underline{j}}}
    = \lambda^j e^{\lambda-\rbra{80L-j}\log\rbra{2}}
    \leq \lambda^j e^{-20L}$, the third inequality uses $\lambda \leq 20L$, and the fourth inequality is by \cref{lemma:poly-approx-xalpha}.
    Combining \cref{eq:case1.1-abs-bias,eq:case1.1-abs-diff-TX,eq:case1.1-abs-TX-large} gives 
    \begin{align*}
        \abs*{ \E\sbra{n^\alpha \mathsf{L}_n\rbra{X}} - \lambda^\alpha }
        & \leq c_\alpha \rbra{40L}^{\alpha} \rbra*{K^{-2\alpha} + 8^K e^{-20L}} \\
        & \leq c_\alpha 40^\alpha \rbra*{\rbra{64\log\rbra{2}}^{2\alpha} + 8\rbra*{\frac{2\alpha}{19e}}^{2\alpha}} L^{-\alpha}
        \leq \frac{C_\alpha}{4} L^{-\alpha},
    \end{align*}
    where we use $K \geq L/\rbra{64\log\rbra{2}}$, $8^K \leq 8e^{3L/32}$, and $L^{2\alpha}e^{-19L} \leq \rbra{2\alpha/\rbra{19e}}^{2\alpha}$.

    \textbf{Case 1.2: $\lambda > 20L$.}
    Since $X$ and $Y$ are independent, 
    \begin{align*}
        \abs*{ \E\sbra{ \mathbf{1}_{\cbra{Y \leq 10L}} \rbra{n^\alpha \mathsf{L}_n\rbra{X} - \lambda^\alpha} } }
        & = \Pr\sbra{Y \leq 10L} \abs*{ \E\sbra{n^\alpha \mathsf{L}_n\rbra{X} - \lambda^\alpha} } \\
        & \leq e^{-\lambda/8} \rbra*{ 16 c_\alpha \rbra{40L}^\alpha e^{L/8} + \lambda^\alpha } \\
        & \leq 16c_\alpha 40^\alpha L^\alpha e^{-19L/8} + \lambda^\alpha e^{-\lambda/8} \\
        & \leq \rbra*{16c_\alpha 40^\alpha \rbra*{\frac{16\alpha}{19e}}^{2\alpha}
            + 20^{-\alpha}\rbra*{\frac{16\alpha}{e}}^{2\alpha}} L^{-\alpha}
        \leq \frac{C_\alpha}{4} L^{-\alpha},
    \end{align*}
    where we use $\Pr\sbra{Y \leq 10L} \leq \Pr\sbra{Y \leq \lambda/2} \leq e^{-\lambda/8}$ by \cref{eq:Poi-low} in \cref{lemma:poi-neq} and $\abs{n^{\alpha} \mathsf{L}_n\rbra{X}}
    \leq \rbra{40L}^\alpha \sum_{j=0}^K \abs{a_j} 2^j
    \leq c_\alpha \rbra{40L}^{\alpha} 16^K
    \leq 16c_\alpha \rbra{40L}^{\alpha} e^{L/8}$
    by \cref{lemma:poly-approx-xalpha}.
    
    Combining Cases 1.1 and 1.2 and using independence gives
    \begin{equation} \label{eq:case1-abs-bias}
    \abs*{ \E\sbra{ \mathbf{1}_{\cbra{Y \leq 10L}} \rbra{n^\alpha \mathsf{L}_n\rbra{X} - \lambda^\alpha} } }
    \leq \frac{C_\alpha}{4} L^{-\alpha}.
    \end{equation}
    For the second moment, using $\rbra{a-b}^2 \leq 2a^2+2b^2$ and the same bound on $\abs{n^\alpha\mathsf{L}_n\rbra{X}}$, we have
    \begin{align}
        \E\sbra{ \mathbf{1}_{\cbra{Y \leq 10L}} \rbra{n^\alpha \mathsf{L}_n\rbra{X} - \lambda^\alpha}^2 } 
        & \leq \begin{cases}
            2c_\alpha^2\rbra{40L}^{2\alpha}256^K + 2\rbra{20L}^{2\alpha}, & \lambda \leq 20L, \\
            \rbra*{2c_\alpha^2\rbra{40L}^{2\alpha}256^K + 2\lambda^{2\alpha}}e^{-\lambda/8}, & \lambda > 20L,
        \end{cases} \nonumber \\
        & \leq \frac{C_\alpha^2}{4} L^{2\alpha}d^{1/4}, \label{eq:case1-second-moment}
    \end{align}
    where we use $256^K \leq 256 e^{L/4} \leq 512 d^{1/4}$.

    \textbf{Case 2: $Y > 10L$.}
    We consider the term $\mathbf{1}_{\cbra{Y > 10L}} \mathsf{H}_n\rbra{X,Y}$.
    When $\lambda=0$, this case has probability zero, so we assume $\lambda>0$ below.

    \textbf{Case 2.1: $\abs{Y-\lambda} \leq \lambda/2$.}
    In this case, $\lambda > 20L/3 > 1$.
    Since $X$ and $Y$ are independent and $\E\sbra{X^{\underline r}}=\lambda^r$, direct calculation gives
    \begin{align}
        \E\sbra*{n^\alpha\mathsf{H}_n\rbra{X,Y} \mid Y}
        & = \sum_{j=0}^J \binom{\alpha}{j} \sum_{r=0}^j \binom{j}{r} \rbra{-1}^{j-r}Y^{\alpha-r}\lambda^r \nonumber \\
        & = \sum_{j=0}^J \binom{\alpha}{j} Y^{\alpha-j}\rbra{\lambda-Y}^j. \label{eq:case2.1-mean}
    \end{align}
    Taylor's theorem applied to $x^\alpha$ at $Y$ gives
    \begin{align}
        \abs*{ \E\sbra*{n^\alpha\mathsf{H}_n\rbra{X,Y} \mid Y} - \lambda^\alpha }
        & \leq \abs*{\binom{\alpha}{J+1}}\rbra{\lambda/2}^{\alpha-J-1}\abs{Y-\lambda}^{J+1} \nonumber \\
        & \leq \rbra{16J}^{2J}\lambda^{\alpha-J-1}\abs{Y-\lambda}^{J+1}, \label{eq:case2.1-Taylor}
    \end{align}
    where we use $\alpha<J/4$ and $\abs{\binom{\alpha}{j}} \leq J^j/j!$ for $0 \leq j \leq J+1$.
    Taking expectations gives
    \begin{align}
        \abs*{ \E\sbra{ \mathbf{1}_{\cbra{Y > 10L,\ \abs{Y-\lambda} \leq \lambda/2}} \rbra{n^\alpha\mathsf{H}_n\rbra{X,Y}-\lambda^\alpha} } } \nonumber & \leq \rbra{16J}^{2J}\lambda^{\alpha-J-1}\E\sbra{\abs{Y-\lambda}^{J+1}} \nonumber \\
        & \leq \rbra{16J}^{2J}\rbra{2\rbra{J+1}}^{J+1}\lambda^{\alpha-\rbra{J+1}/2} \nonumber \\
        &
        \leq \rbra{16J}^{4J}L^{-\alpha}
        \leq \frac{C_\alpha}{4}L^{-\alpha}, \label{eq:case2.1-abs-bias}
    \end{align}
    where we use $J+1 \geq 4\alpha+4$, $\lambda>L$, and \cref{eq:EX-lambda-q} in \cref{lemma:poi-neq}.

    For the second moment, expanding the square in \cref{eq:def-Hn} gives
    \begin{align}
        \Var\sbra*{n^\alpha\mathsf{H}_n\rbra{X,Y} \mid Y}
        & = \sum_{r=1}^J \frac{\lambda^r}{r!}
            \rbra*{\sum_{j=r}^J \binom{\alpha}{j}j^{\underline r}Y^{\alpha-j}\rbra{\lambda-Y}^{j-r}}^2 \nonumber \\
        & \leq \rbra{16J}^{4J}\sum_{r=1}^J\frac{\lambda^{2\alpha-r}}{r!}
        \leq e\rbra{16J}^{4J}\lambda^{2\alpha-1}. \label{eq:case2.1-var}
    \end{align}
    Here, the equality follows from the Poisson factorial moment identity
    \[
    \E\sbra{X^{\underline j}X^{\underline r}}
    = \sum_{i=0}^{\min\cbra{j,r}}\binom{j}{i}\binom{r}{i}i!\lambda^{j+r-i},
    \]
    and the first inequality uses, for $1 \leq r \leq J$,
    \begin{align*}
        \abs*{\sum_{j=r}^J \binom{\alpha}{j}j^{\underline r}Y^{\alpha-j}\rbra{\lambda-Y}^{j-r}}
        \leq \lambda^{\alpha-r}\sum_{j=r}^J\frac{J^j}{j!}J^r 2^{\abs{\alpha-j}}2^{r-j} 
        \leq 2^{2J}J^Je^{J/2}\lambda^{\alpha-r}
        \leq \rbra{16J}^{2J}\lambda^{\alpha-r}.
    \end{align*}
    Combining \cref{eq:case2.1-Taylor,eq:case2.1-var} and using \cref{eq:EX-lambda-q}, we have
    \begin{align}
        \E\sbra{ \mathbf{1}_{\cbra{Y > 10L,\ \abs{Y-\lambda} \leq \lambda/2}} \rbra{n^\alpha\mathsf{H}_n\rbra{X,Y}-\lambda^\alpha}^2 } \nonumber & \leq \rbra{16J}^{4J}\rbra*{e\lambda^{2\alpha-1}
            + \lambda^{2\alpha-2J-2}\E\sbra{\abs{Y-\lambda}^{2J+2}}} \nonumber \\
        & \leq \rbra{16J}^{4J}\rbra*{e\lambda^{2\alpha-1}
            + \rbra{4\rbra{J+1}}^{2J+2}\lambda^{2\alpha-J-1}} \nonumber \\
        &
        \leq \rbra{16J}^{10J}\lambda^{2\alpha-1}, \label{eq:case2.1-second-moment}
    \end{align}
    where we use \cref{eq:EX-lambda-q} in \cref{lemma:poi-neq}.
    Here, note that for $0 < \alpha \leq 1/2$, $\lambda^{2\alpha-1} \leq 1$; for $\alpha>1/2$, $n^{-2\alpha}\lambda^{2\alpha-1}=p^{2\alpha-1}/n$.

    \textbf{Case 2.2: $\abs{Y-\lambda} > \lambda/2$.}
    The Poisson tail bounds give
    \begin{align}
        \Pr\sbra{Y > 10L,\ \abs{Y-\lambda} > \lambda/2}
        & \leq \begin{cases}
            e^{\lambda-10L\log\rbra{2}}, & 0<\lambda<5L, \\
            2e^{-\lambda/12}, & \lambda \geq 5L,
        \end{cases} \nonumber \\
        & \leq 2e^{-\rbra{\lambda+L}/16}, \label{eq:case2.2-prob}
    \end{align}
    where the first case follows from $\Pr\sbra{Y>10L} \leq 2^{-10L}\E\sbra{2^Y}$, and the second case follows from \cref{eq:Poi-upp,eq:Poi-low} in \cref{lemma:poi-neq} with $u=v=1/2$.
    By \cref{eq:def-Hn},
    \begin{align*}
        \abs{n^\alpha\mathsf{H}_n\rbra{X,Y}}
        \leq \sum_{j=0}^J\abs*{\binom{\alpha}{j}}\sum_{r=0}^j\binom{j}{r}Y^{\alpha-r}X^{\underline r} 
        \leq \rbra{1+X+Y}^{2J}\sum_{j=0}^J\abs*{\binom{\alpha}{j}}2^j
        \leq 2^{6J}\rbra{1+X+Y}^{2J},
    \end{align*}
    where we use $Y>10L>1$, $\alpha<J$, $X^{\underline r} \leq X^r$, and
    $\sum_{j=0}^J\abs{\binom{\alpha}{j}}2^j \leq \sum_{j=0}^J\rbra{2J}^j/j! \leq e^{2J} \leq 2^{6J}$.
    Therefore,
    \[
    \abs{n^\alpha\mathsf{H}_n\rbra{X,Y}-\lambda^\alpha}
    \leq 2^{6J+1}\rbra{1+X+Y+\lambda}^{2J}.
    \]
    By the Cauchy--Schwarz inequality,  \cref{eq:XY-Poi} in \cref{lemma:poi-neq}, and \cref{eq:case2.2-prob}, we have
    \begin{align}
        \abs*{ \E\sbra{ \mathbf{1}_{\cbra{Y > 10L,\ \abs{Y-\lambda} > \lambda/2}} \rbra{n^\alpha\mathsf{H}_n\rbra{X,Y}-\lambda^\alpha} } } \nonumber & \leq 2^{6J+1}\sqrt{\E\sbra{\rbra{1+X+Y+\lambda}^{4J}}}
            \sqrt{2e^{-\rbra{\lambda+L}/16}} \nonumber \\
        & \leq 2^{6J+2}\rbra{16J}^{2J}\rbra{1+\lambda}^{2J}e^{-\rbra{\lambda+L}/32}
        \leq \frac{C_\alpha}{4}L^{-\alpha}, \label{eq:case2.2-abs-bias}
    \end{align}
    and similarly,
    \begin{align}
        \E\sbra{ \mathbf{1}_{\cbra{Y > 10L,\ \abs{Y-\lambda} > \lambda/2}} \rbra{n^\alpha\mathsf{H}_n\rbra{X,Y}-\lambda^\alpha}^2 } \nonumber & \leq 2^{12J+2}\sqrt{\E\sbra{\rbra{1+X+Y+\lambda}^{8J}}}
            \sqrt{2e^{-\rbra{\lambda+L}/16}} \nonumber \\
        & \leq 2^{12J+3}\rbra{32J}^{4J}\rbra{1+\lambda}^{4J}e^{-\rbra{\lambda+L}/32}
        \leq \frac{C_\alpha^2}{4}. \label{eq:case2.2-second-moment}
    \end{align}

    \textbf{Putting it all together.} Combining \cref{eq:case1-abs-bias,eq:case2.1-abs-bias,eq:case2.2-abs-bias} and using $\lambda=np$ gives
    \[
    \abs*{\E\sbra{Z}-p^\alpha}
    \leq \frac{3C_\alpha}{4}\rbra{nL}^{-\alpha}
    \leq C_\alpha\rbra{nL}^{-\alpha},
    \]
    which proves \cref{eq:cond1-abs-bias}.
    On the other hand, combining \cref{eq:case1-second-moment,eq:case2.1-second-moment,eq:case2.2-second-moment} gives
    \begin{align*}
        \Var\sbra{Z}
        & \leq \E\sbra{\rbra{Z-p^\alpha}^2} \\
        & \leq \frac{3C_\alpha^2}{4}\frac{L^{2\alpha}d^{1/4}}{n^{2\alpha}}
            + \frac{C_\alpha^2}{4}\mathbf{1}_{\cbra{\alpha>1/2}}\frac{p^{2\alpha-1}}{n} \\
        & \leq C_\alpha^2\rbra*{\frac{L^{2\alpha}d^{1/4}}{n^{2\alpha}}
            + \mathbf{1}_{\cbra{\alpha>1/2}}\frac{p^{2\alpha-1}}{n}},
    \end{align*}
    which proves \cref{eq:cond2-var}.
\end{proof}

\begin{lemma}
    Let $\widehat{F}$ be the estimator defined in \cref{eq:def-estimator}.
    Then,
    \begin{align} \label{eq:relative-error}
        \E\sbra*{ \rbra*{ \frac{ \widehat{F} - \mathrm{F}_\alpha\rbra{p} }{ \mathrm{F}_\alpha\rbra{p} } }^2 }
        \leq C_\alpha^4 \rbra*{ \frac{d^{\max\cbra{2, 2\alpha}}}{\rbra{nL}^{2\alpha}} + \mathbf{1}_{\cbra{\alpha>1/2}} \frac{d^{\abs{1-1/\alpha}}}{n} },
    \end{align}
    where $C_\alpha$ is the constant defined in \cref{lemma:Poi-single-error}.
\end{lemma}

\begin{proof}
    By \cref{lemma:Poi-single-error}, we have
    \begin{align*}
        \rbra*{ \E\sbra{ \widehat{F} } - \mathrm{F}_\alpha\rbra{p} }^2 & \leq C_\alpha^2 d^2 \rbra{nL}^{-2\alpha}, \\
        \Var\sbra{\widehat{F}} & \leq C_\alpha^2\rbra*{\frac{d^{5/4}L^{2\alpha}}{n^{2\alpha}} + \frac{\mathbf{1}_{\cbra{\alpha>1/2}}}{n} \sum_{i \in \sbra{d} \colon p_i > 0} p_i^{2\alpha-1}}.
    \end{align*}
    Therefore,
    \begin{align}
        \E\sbra*{ \rbra*{ \frac{ \widehat{F}  - \mathrm{F}_\alpha\rbra{p} }{ \mathrm{F}_\alpha\rbra{p} } }^2 }
        & = \E\sbra*{ \frac{ \rbra{\E\sbra{ \widehat{F} } - \mathrm{F}_\alpha\rbra{p}}^2 }{ \rbra{\mathrm{F}_\alpha\rbra{p}}^2 } + \frac{\Var\sbra{\widehat{F}}}{ \rbra{\mathrm{F}_\alpha\rbra{p}}^2 } } \\
        & \leq C_\alpha^2 \rbra*{ \frac{d^2}{\rbra{nL}^{2\alpha} \rbra{\mathrm{F}_\alpha\rbra{p}}^2} + \frac{d^{5/4} L^{2\alpha}}{n^{2\alpha} \rbra{\mathrm{F}_\alpha\rbra{p}}^2} +  \frac{\mathbf{1}_{\cbra{\alpha>1/2}}}{n \rbra{\mathrm{F}_\alpha\rbra{p}}^2} \sum_{i \in \sbra{d} \colon p_i > 0} p_i^{2\alpha-1} } \\
        & = C_\alpha^2 \rbra*{ \frac{d^2}{\rbra{nL}^{2\alpha} \rbra{\mathrm{F}_\alpha\rbra{p}}^2} \rbra*{1 + \frac{L^{4\alpha}}{d^{3/4}}} +  \frac{\mathbf{1}_{\cbra{\alpha>1/2}}}{n \rbra{\mathrm{F}_\alpha\rbra{p}}^2} \sum_{i \in \sbra{d} \colon p_i > 0} p_i^{2\alpha-1} },
    \end{align}
    which gives \cref{eq:relative-error} using $1 + L^{4\alpha}/d^{3/4} \leq 1+e^{3/4}\rbra{16\alpha/3e}^{4\alpha} \leq C_\alpha^2$, $d^2 / \rbra{\mathrm{F}_\alpha\rbra{p}}^2 \leq d^{\max\cbra{2, 2\alpha}}$ and $\mathrm{F}_{2\alpha-1}\rbra{p} / \rbra{\mathrm{F}_\alpha\rbra{p}}^2 \leq d^{\abs{1-1/\alpha}}$ for $\alpha > 1/2$ and $\alpha \neq 1$ given in \cref{lemma:2alpha-1-bound}.
\end{proof}

\begin{proof}[Proof of \cref{thm:estimator}]
    By taking 
    \[
    n = \ceil*{ C^{1/\rbra{2\alpha}} \frac{d^{\max\cbra{1, 1/\alpha}}}{\varepsilon^{1/\alpha} L} + C \mathbf{1}_{\cbra{\alpha > 1/2}} \frac{d^{\abs{1-1/\alpha}}}{\varepsilon^2} }, \qquad C = \frac{24C_\alpha^4}{\min\cbra{1, \abs{1-\alpha}^2}},
    \]
    where $C_\alpha$ is the constant defined in \cref{lemma:Poi-single-error}.
    Markov's inequality gives 
    \[
    \Pr\sbra*{\frac{\abs{\widehat{F} - \mathrm{F}_\alpha\rbra{p}}}{\mathrm{F}_\alpha\rbra{p}} \leq \frac{\min\cbra{1, \abs{1-\alpha}}}{2} \varepsilon} \geq 1 - \frac{4}{\min\cbra{1, \abs{1-\alpha}^2} \varepsilon^2} \E\sbra*{ \rbra*{ \frac{ \widehat{F} - \mathrm{F}_\alpha\rbra{p} }{ \mathrm{F}_\alpha\rbra{p} } }^2 } \geq \frac{2}{3},
    \]
    in which case, 
    \[
    \abs{\widehat{H} - \mathrm{H}_\alpha\rbra{p}} \leq \frac{2}{\abs{1-\alpha}} \frac{\abs{\operatorname{clip}_{I_\alpha}\rbra{\widehat{F}} - \mathrm{F}_\alpha\rbra{p}}}{\mathrm{F}_\alpha\rbra{p}} \leq \frac{2}{\abs{1-\alpha}} \frac{\abs{\widehat{F} - \mathrm{F}_\alpha\rbra{p}}}{\mathrm{F}_\alpha\rbra{p}} \leq \varepsilon,
    \]
    where we use $\abs{\log\rbra{1+x}}\leq2\abs{x}$ for $\abs{x}\leq1/2$.
\end{proof}

\section*{Acknowledgment}

The author used Large Language Models as AI-assisted research and writing tools throughout the
preparation of this paper. The tools were used to help brainstorm ideas and explore proof
strategies. Portions of the paper text were redrafted or modified with AI assistance across all
sections. All final mathematical claims, algorithms, proofs, citations, and wording were reviewed,
edited, and validated by the author. The author assumes responsibility for all content of the paper.

\addcontentsline{toc}{section}{References}

\bibliographystyle{alphaurl}
\bibliography{main}

\end{document}